\documentclass[12pt]{article}

\usepackage[a4paper,margin=1in]{geometry}
\usepackage{amsmath,amssymb,amsfonts,amsthm, bm}
\usepackage{mathtools}
\usepackage{mathrsfs}
\usepackage{graphicx}
\usepackage{tikz}
\usetikzlibrary{arrows.meta,calc,decorations.markings}
\usetikzlibrary{positioning}
\usepackage{standalone}
\usepackage{hyperref}
\hypersetup{
    colorlinks=true,
    linkcolor=blue,
    citecolor=blue,
    urlcolor=blue
}
\usepackage{enumitem}
\usepackage{cite}

\theoremstyle{plain}
\newtheorem{theorem}{Theorem}[section]
\newtheorem{proposition}[theorem]{Proposition}
\newtheorem{lemma}[theorem]{Lemma}
\newtheorem{corollary}[theorem]{Corollary}

\theoremstyle{definition}
\newtheorem{definition}[theorem]{Definition}

\newtheorem{remark}[theorem]{Remark}

\newcommand{\R}{\mathbb{R}}

\usepackage{fullpage}
\title{Collision geometry of relativistic spinning particles}

\author{Simone Calogero}

\begin{document}

\maketitle

\begin{abstract}
We investigate elastic binary collisions of relativistic spinning particles in special relativity. The spin of each particle is represented by an antisymmetric second order tensor. Assuming the conservation of total four momentum and total spin tensor, together with the mass shell and spin constraints, we formulate the collision problem in a fully Lorentz covariant setting.
We show that the relativistic collision problem admits a simple geometric formulation, reducing the original system of conservation laws to the solution of a quadratic equation on a circle. This reduction yields a complete classification of the postcollisional states together with explicit reconstruction formulas for all postcollisional variables from the conserved quantities. In particular, there are generically only finitely many postcollisional states, with the maximal number equal to eight. 
\end{abstract}


\section{Introduction}\label{intro}
Binary collisions constitute the fundamental microscopic mechanism governing the dynamics of dilute gases and play a central role in kinetic theory. In both the classical and relativistic settings, the collision operator is entirely determined by the relation between the precollisional and postcollisional variables of two interacting particles. Consequently, understanding the geometry of binary collisions is a prerequisite for the construction and analysis of kinetic models~\cite{Cerci}.

For identical relativistic particles without spin, the collision problem is completely determined by the conservation of total four-momentum together with the mass shell constraints. For a fixed set of precollisional momenta, the corresponding postcollisional momenta form a two-dimensional set parametrized by a unit vector
representing the scattering direction in the center of momentum frame~\cite{Groot, Glassey,Strain}. This parametrization underlies the relativistic Boltzmann collision operator and provides the geometric description of elastic relativistic collisions in the absence of internal degrees of freedom.

The inclusion of spin substantially modifies this picture. Besides the conservation of total four-momentum, one must also impose the conservation of total spin,  together with the mass shell and spin constraints. Throughout this paper we restrict our attention to elastic binary collisions of identical relativistic particles. Accordingly, both particles are assumed to have the same rest mass $m>0$ and the same fixed spin magnitude $\sigma>0$. Following Frenkel~\cite{Frenkel}, the spin of each particle is described by an antisymmetric second order tensor satisfying a system of orthogonality and normalization conditions. 

The additional conservation law of the total spin tensor makes the determination of the postcollisional states considerably more involved than in the spinless case. Indeed, while the postcollisional momenta alone are still parametrized by the scattering direction, the postcollisional spins must simultaneously satisfy the conservation of the total spin tensor together with a system of nonlinear algebraic constraints involving both the postcollisional momenta and the postcollisional spins.

The purpose of the present paper is to provide a complete geometric analysis of this collision problem. We first formulate the collision equations in a fully Lorentz covariant manner. We then show that, for identical particles, the conservation of the total spin tensor may be equivalently expressed in terms of a conserved bivector constructed from the particle momenta and particle spins. Exploiting Lorentz invariance, the problem is reduced to the center of momentum frame, where the postcollisional variables are shown to be completely determined by a scattering direction and a single scalar parameter describing the longitudinal component of the postcollisional spin difference.

Our main result is that the collision problem is equivalent to the solution of a single quadratic equation on a circle. This reduction leads to a complete classification of the postcollisional states together with explicit reconstruction formulas for the postcollisional momenta and spins. In particular, whereas the spinless relativistic collision problem admits a two dimensional continuum of postcollisional states, the spinning collision problem generically possesses only a finite even number of postcollisional states, up to a maximum of eight, while an exceptional class of collisions admits infinitely many postcollisional states. 

The collision problem considered in this paper was first studied by John Lighton Synge in~\cite{Synge}, where he derived some special solutions and observed that, generically, there can be at most eight postcollisional configurations. In the present paper we provide a complete solution to Synge's collision problem for relativistic spinning particles.

The paper is organized as follows. In Section~\ref{formulation} we formulate the collision problem in a Lorentz covariant framework and reduce it to the center of momentum frame. Section~\ref{geometry} contains the geometric analysis of the collision manifold, including its complete classification. In Section~\ref{reconstructionSEC} we give explicit reconstruction formulas for the admissible collisional states. Finally, Section~\ref{kinetic} discusses the implications of the collision geometry for kinetic theory and introduces the notion of a collision kernel.

\subsection*{Notation}
Throughout the paper, we work on Minkowski spacetime, i.e., the manifold $\mathbb{R}^{4}$ endowed with the Minkowski metric $\eta$. In the standard Cartesian coordinates,
\[
\eta_{\mu\nu}=\operatorname{diag}(-1,1,1,1).
\]
Greek indices $\mu,\nu,\alpha,\beta,\ldots$ take values in $\{0,1,2,3\}$ and are raised and lowered using the metric $\eta$. Repeated indices are summed over according to Einstein's summation convention. 
For any two four-vectors $a$ and $b$, we denote by $a\wedge b$ the antisymmetric tensor with components
\[
(a\wedge b)^{\mu\nu}
=
a^\mu b^\nu-a^\nu b^\mu.
\]
Moreover, $\varepsilon^{\mu\nu\alpha\beta}$ denotes the totally antisymmetric Levi-Civita tensor normalized by
\[
\varepsilon^{0123}=1.
\]
Finally, the Hodge dual of an antisymmetric tensor $A$ is defined by 
\[ 
(*A)^{\mu\nu} = \frac12 \varepsilon^{\mu\nu\alpha\beta} A_{\alpha\beta}, 
\] 
so that $*^2=-\operatorname{Id}$ on bivectors.

\section{Formulation of the collision problem}\label{formulation}
Consider a relativistic spinning particle with rest mass $m>0$ and four-momentum $p=(p^0,{\bm p})\in\R^4$ satisfying
\begin{equation}\label{massshell}
p^\mu p_\mu =-m^2,\quad p^0>0.
\end{equation}
Following Frenkel~\cite{Frenkel}, we describe the spin by an antisymmetric tensor $\phi\in \Lambda^2\R^4$ such that 
\begin{equation}\label{frenkelcond}
\phi^{\mu\nu}p_\mu=0,\quad
\phi^{\mu\nu}\phi_{\mu\nu}=\sigma^2,
\end{equation}
where $\sigma>0$ is a (dimensionless) constant representing the spin magnitude. Defining the spin vector $s\in\R^4$ by
\begin{equation}\label{sfromphi}
s^\mu=\frac{\sqrt{2}}{m}(*\phi)^{\mu\nu}p_\nu,
\end{equation}
we find 
\begin{equation}\label{phifroms}
\phi = \frac{1}{\sqrt{2}m}(*p\wedge s),
\end{equation} 
as well as
\begin{equation}\label{thomascond}
s^\mu p_\mu =0,\quad s^\mu s_\mu =\sigma^2.
\end{equation}
Conversely, if a vector $s\in\R^4$ satisfies \eqref{thomascond}, then the two-form \eqref{phifroms} satisfies Frenkel's conditions \eqref{frenkelcond}.
It follows that the manifolds
\begin{align*}
&\Gamma:=\{(p,\phi)\in\R^4\times\Lambda^2\R^4\,:\,\text{\eqref{massshell} and \eqref{frenkelcond} hold}\}\\
&\Omega:=\{(p,s)\in\R^4\times\R^4\,:\,\text{\eqref{massshell} and \eqref{thomascond} hold}\}
\end{align*}
are isomorphic. The manifolds $\Gamma$ and $\Omega$ represent the state space of a single relativistic spinning particle in the spin tensor and spin vector representations, respectively.
 
Consider now a second particle with rest mass $m$, spin magnitude $\sigma$, four-momentum $q$, spin tensor $\psi$ and spin vector $r$ such that $(q,\psi)\in\Gamma$---or equivalently $(q,r)\in\Omega$. The total four-momentum and total spin tensor of the two particle system are defined by $p+q$ and $\phi+\psi$, respectively.
Suppose that the particles collide at some future time. Then we can refer to
\[
\langle p,\phi;q,\psi\rangle:=((p,\phi),(q,\psi))\in\Gamma^2,
\] 
or equivalently to
\[
\langle p,s;q,r\rangle:=((p,s),(q,r))\in\Omega^2,
\] 
as an admissible precollisional state. Let us denote by $p',q',\phi',\psi'$, or equivalently by $p',q',s',r'$, the postcollisional state variables.
\begin{definition}
The collision of two particles in the admissible precollisional state $\langle p,\phi;q,\psi\rangle\in\Gamma^2$ is said to be elastic if the postcollisional state is admissible, i.e.,  
\[
\langle p',\phi';q',\psi'\rangle\in\Gamma^2,
\] 
and
\begin{align*}
&p'+q'=p+q\quad \text{(conservation of total four-momentum),}\\
&\phi'+\psi'=\phi+\psi\quad  \text{(conservation of total spin tensor).}
\end{align*}
\end{definition}

Since the particles have identical masses, the normalization factor in~\eqref{phifroms} is the same for both particles. Hence, the conservation of the total spin tensor is equivalent to  
\[
p'\wedge s'+q'\wedge r'=p\wedge s+q\wedge r.
\]
This equation expresses the conservation of the bivector
\[
K:=p\wedge s+q\wedge r.
\]
It is therefore natural to introduce
\begin{equation}\label{P,K}
P:=p+q,\quad
K:=p\wedge s+q\wedge r,
\end{equation}
and formulate the relativistic collision problem as follows.

{\bf Covariant collision problem.} 
Given an admissible precollisional state $\langle p,s;q,r\rangle\in\Omega^2$, 
determine all admissible postcollisional states $\langle p',s';q',r'\rangle\in\Omega^2$
satisfying
\[
p'+q'=P,
\qquad
p'\wedge s'+q'\wedge r'=K
\]
where $P,K$ are given by~\eqref{P,K}.

As we shall see, the entire analysis of the collision problem reduces to the geometric study of the conserved pair $(P,K)$.

\subsection{Threshold collision problem}
We continue this section by solving the collision problem in a special, threshold case.
\begin{definition}
The elastic collision of two particles in an admissible precollisional state is said to be at threshold if
\[
p=q.
\]
\end{definition}

\begin{lemma}\label{simple}
The threshold condition $p=q$ is equivalent to
\[
P^\mu P_\mu=-4m^2.
\]
\end{lemma}
\begin{proof}
Since
\[
p^\mu p_\mu=q^\mu q_\mu=-m^2,
\]
we have
\[
P^\mu P_\mu=(p+q)^\mu(p+q)_\mu
=-2m^2+2p^\mu q_\mu.
\]
Thus, the condition $P^\mu P_\mu=-4m^2$ is equivalent to  $p^\mu q_\mu=-m^2$
Hence,
\[
(p-q)^\mu(p-q)_\mu
=
p^\mu p_\mu+q^\mu q_\mu-2p^\mu q_\mu
=0.
\]
Since $p$ and $q$ are future directed timelike vectors of the same mass, this is equivalent to the threshold condition $p=q$.
\end{proof}
\begin{proposition}\label{thresholdcase}
Let $\langle p,\phi;q,\psi\rangle\in\Gamma^2$ be an admissible precollisional state. 
Assume that the collision is at threshold.
Then the pre- and postcollisional four-momenta satisfy
\[
p'=q'=p=q.
\] 
Let $\Phi:=\phi+\psi$. 
Then the set of admissible postcollisional spin tensors is given by
\[
\phi'=\frac{\Phi}{2}+\chi,
\qquad
\psi'=\frac{\Phi}{2}-\chi,
\]
where $\chi\in\Lambda^2\mathbb R^4$ satisfies
\[
\chi^{\mu\nu}p_\mu=0,\quad 
\Phi^{\mu\nu}\chi_{\mu\nu}=0,
\quad
\chi^{\mu\nu}\chi_{\mu\nu}=\sigma^2-\frac14\Phi^{\mu\nu}\Phi_{\mu\nu}.
\]
Equivalently, in the spin four-vector representation, if $S:=s+r$,
then all admissible postcollisional spin four-vectors are given by
\[
s'=\frac{S}{2}+\xi,
\qquad
r'=\frac{S}{2}-\xi,
\]
where $\xi\in\mathbb R^4$ satisfies
\[
\xi^\mu p_\mu=0,\quad
\xi^\mu S_\mu=0,\quad 
\xi^\mu\xi_\mu=\sigma^2-\frac14 S^\mu S_\mu.
\]
The set of admissible postcollisional states reduces to a single element if and only if $\phi=\psi$,
or equivalently, $s=r$. In this case
\begin{equation}\label{physical}
\phi'=\phi,\quad \psi'=\psi,\quad\text{or, equivalently,}\quad
s'=s,\quad r'=r
\end{equation}
Otherwise, the admissible postcollisional states form a one-parameter family.
\end{proposition}

\begin{proof}
As $P^\mu P_\mu$ is a collision invariant, the identities $p'=q'=p=q$ follow by Lemma~\ref{simple}. 
It remains to describe the spin tensors. Conservation of the total spin tensor gives
\[
\phi'+\psi'=\Phi.
\]
Therefore there exists an antisymmetric tensor $\chi$ such that
\[
\phi'=\frac{\Phi}{2}+\chi,
\qquad
\psi'=\frac{\Phi}{2}-\chi.
\]
Since $p'=q'=p$, the orthogonality conditions for $\phi'$ and $\psi'$ are equivalent to $\chi^{\mu\nu}p_\mu=0$. 
The spin normalization constraints give
\[
\left(\frac{\Phi}{2}+\chi\right)^{\mu\nu}
\left(\frac{\Phi}{2}+\chi\right)_{\mu\nu}
=
\sigma^2
\]
and
\[
\left(\frac{\Phi}{2}-\chi\right)^{\mu\nu}
\left(\frac{\Phi}{2}-\chi\right)_{\mu\nu}
=
\sigma^2.
\]
Subtracting the two equations yields
\[
\Phi^{\mu\nu}\chi_{\mu\nu}=0.
\]
Adding them yields
\[
\chi^{\mu\nu}\chi_{\mu\nu}
=
\sigma^2-\frac14\Phi^{\mu\nu}\Phi_{\mu\nu}.
\]
This proves the tensor representation.
The spin four-vector representation follows in the same way. At threshold $p=q$, and conservation of the bivector $K$ is equivalent to
\[
p\wedge(s'+r')=p\wedge(s+r).
\]
Since all spin four-vectors are orthogonal to $p$, this implies $s'+r'=S$.
Thus,
\[
s'=\frac{S}{2}+\xi,
\qquad
r'=\frac{S}{2}-\xi.
\]
The constraints $s'^\mu p_\mu=r'^\mu p_\mu=0$ give $\xi^\mu p_\mu=0$,
while the equal spin magnitudes give
\[
\xi^\mu S_\mu=0,\quad
\xi^\mu\xi_\mu
=
\sigma^2-\frac14 S^\mu S_\mu:=\rho
\]
The equations $\xi^\mu p_\mu=0$ and $\xi^\mu S_\mu=0$ define a two dimensional spacelike plane, and $\xi^\mu \xi_\mu=\rho$ restricts $\xi$ to a circle. The admissible postcollisional states form a one-parameter family whenever $\rho>0$. The only exception occurs when the circle degenerates to a single point, that is, when $\rho=0$. This is equivalent to $\rho=0$, i.e., $S^\mu S_\mu =4\sigma^2$. Since $s^\mu s_\mu=r^\mu r_\mu=\sigma^2$, this occurs if and only if $s=r$, equivalently $\phi=\psi$. 
\end{proof}
\begin{remark}
The threshold case corresponds to two particles having identical four-momenta. It should therefore be regarded as a degenerate limiting case of the collision problem rather than a genuine scattering event. In this situation, the physically natural solution is~\eqref{physical}, since it reflects the indistinguishability of two particles occupying the same state.
\end{remark}

\subsection{Reduction to the center of momentum frame}
Our next goal is to formulate the non-threshold collision problem in the center of momentum frame. 
Let $\Lambda\in O(1,3)$ be a proper orthochronous Lorentz transformation. Given an admissible state
\[
\langle p,s;q,r\rangle\in\Omega^2,
\]
we define
\[
\Lambda\langle p,s;q,r\rangle:=\langle\Lambda p,\Lambda s;\Lambda q,\Lambda r\rangle.
\]

\begin{proposition}
The set $\Omega^2$ is invariant under proper orthochronous Lorentz transformations. That is, 
\[
\langle p,s;q,r\rangle\in\Omega^2\Rightarrow \Lambda\langle p,s;q,r\rangle \in\Omega^2
\]
Moreover,
\[
P\mapsto\Lambda P,\qquad K\mapsto\Lambda K\Lambda^T,
\]
and the collision problem is Lorentz invariant.
\end{proposition}

\begin{proof}
Since Lorentz transformations preserve the Minkowski metric, 
\[
(\Lambda p)^\mu(\Lambda p)_\mu=p^\mu p_\mu,
\]
and so the mass shell constraint is preserved.
Likewise,
\[
(\Lambda s)^\mu(\Lambda p)_\mu=s^\mu p_\mu,\quad
(\Lambda s)^\mu(\Lambda s)_\mu=s^\mu s_\mu.
\]
Hence, $\Omega$, and therefore also $\Omega^2$, is Lorentz invariant.
Furthermore, $P'=P$ is equivalent to $\Lambda P'=\Lambda P$, 
while $K'=K$ is equivalent to $\Lambda K'\Lambda^T=\Lambda K\Lambda^T$.
Therefore, the collision equations are Lorentz invariant.
\end{proof}

Since $P^\mu P_\mu<0$,
there exists a proper orthochronous Lorentz transformation sending $P$ to $(M,0,0,0)$, where
\[
M=\sqrt{-P^\mu P_\mu}
\]
is the invariant mass of the two-particle system.
The corresponding frame is called the center of momentum (COM) frame~\cite{Groot}.  In the remainder of the paper we work in this frame. It is defined up to spatial rotations---a property which will we exploit in Section~\ref{reconstructionSEC}. Thus,
\[
P=(M,0,0,0)\quad\text{in the COM frame.}
\]
The conservation of momentum implies
\[
p=(M/2,\bm p),\qquad q=(M/2,\bm q),\quad \bm q=-\bm p
\]
with
\begin{equation}\label{kappa}
|\bm p|=|\bm q|=\frac{1}{2}\sqrt{M^2-4m^2}:=\kappa.
\end{equation}
As the threshold case has already been solved in Proposition~\ref{thresholdcase}, we may assume 
\[
\kappa>0.
\]
Without loss of generality we may write
\begin{equation}\label{p}
\bm p=\kappa\,\bm \omega ,\quad \bm q=-\kappa\,\bm\omega \quad \text{where}\quad \bm \omega \in S^2.
\end{equation}
The spin vectors $s=(s^0,{\bm s}), r=(r^0,{\bm r})$ in the COM frame satisfy
\begin{equation}\label{s0}
s^0=\frac{2}{M}\,\bm p\cdot\bm s=\frac{2\kappa}{M}\bm s\cdot \bm \omega,
\qquad
r^0=\frac{2}{M}\,\bm q\cdot\bm r=-\frac{2\kappa}{M}\bm r\cdot \bm \omega,
\end{equation}
so that the independent variables are the spatial vectors $\bm s, \bm r\in\R^3$ and the scattering direction $\bm \omega\in S^2$.

Let $\bm E$, $\bm B$ denote, respectively, the electric and magnetic parts of the bivector $K$; that is,
\[
E^i=K^{0i},\quad B^l=\frac{1}{2}\varepsilon^{ijl}K_{ij}.
\]
The conservation of the bivector $K$ is equivalent to the conservation of the vectors $\bm E$, $\bm B$. 
One computes
\begin{equation}\label{EBpre}
\bm E=\frac{M}{2}{\bm U}-\frac{2\kappa^2}{M}(\bm \omega \cdot \bm U)\,\bm \omega ,\quad \bm B=\kappa\, \bm \omega \times \bm V,
\end{equation}
where we introduced the new independent spin variables
\begin{equation}\label{UV}
\bm U=\bm s+\bm r,\quad \bm V=\bm s-\bm r.
\end{equation}
In terms of the variables ${\bm U}, {\bm V}, {\bm \omega}$, the normalization constraints $s^\mu s_\mu=r^\mu r_\mu =\sigma^2$ are equivalent to
\begin{subequations}\label{spinconstraint}
\begin{align}
&-\frac{4\kappa^2}{M^2}({\bm \omega}\cdot{\bm U})({\bm \omega}\cdot{\bm V})+{\bm U}\cdot{\bm V}=0,\\
&-\frac{2\kappa^2}{M^2}[({\bm \omega}\cdot{\bm U})^2+({\bm \omega}\cdot{\bm V})^2]+\frac{1}{2}(|{\bm U}|^2+|{\bm V}|^2)=2\sigma^2.
\end{align}
\end{subequations}
The first equation in~\eqref{EBpre} implies
\begin{equation}\label{Uprime}
{\bm U}=\frac2M\left({\bm E}-({\bm E}\cdot \bm \omega  )\bm \omega  \right)+\frac{M}{2m^2}({\bm E}\cdot \bm \omega  )\bm \omega  .
\end{equation}
The second equation in~\eqref{EBpre} implies ${\bm B}\cdot \bm \omega  =0$, and yields
\begin{equation}\label{Vprime}
{\bm V}= \frac1\kappa {\bm B}\times \bm \omega  +\alpha \bm \omega  ,
\end{equation}
where
\begin{equation}\label{alpha}
\alpha=\bm \omega  \cdot {\bm V}
\end{equation}
is an arbitrary scalar.
Substituting the found expressions of ${\bm U}$, ${\bm V}$ into the spin normalization constraints~\eqref{spinconstraint} one obtains the equations
\begin{subequations}\label{spinconstraint2}
\begin{align}
&({\bm E}\cdot \bm \omega  )\alpha=\frac1\kappa \bm \omega  \cdot({\bm B}\times {\bm E}),\label{first}\\
&\frac{4|{\bm E}|^2}{M^2}+\frac{4\kappa^2}{M^2m^2}({\bm E}\cdot \bm \omega  )^2+\frac{|{\bm B}|^2}{\kappa^2}+\frac{4m^2}{M^2}\alpha^2=4\sigma^2\label{second},
\end{align}
\end{subequations}
where we recall that $\kappa=\kappa(M)$ is given by~\eqref{kappa}. 
\begin{definition}
Given
\[
M>2m,\qquad{\bm E},{\bm B}\in\mathbb R^3,
\]
the collision manifold is
\[
\mathcal{X}_\mathrm{COM}(M,{\bm E}, {\bm B})=\{(\bm\omega,\alpha)\in S^2\times\R\, :\, {\bm B}\cdot\bm\omega=0\,\text{and~\eqref{spinconstraint2} hold}\}.
\]
\end{definition}

\begin{remark}
For any fixed triple of the conserved quantities $M, {\bm E}, {\bm B}$, each element of $\mathcal{X}_\mathrm{COM}(M, {\bm E}, {\bm B})$ generates an admissible precollisional state. As the conditions in the definition of $\mathcal{X}_\mathrm{COM}(M,{\bm E}, {\bm B})$ have been derived in the COM frame under the mass shell and spin constraints, they also hold for the postcollisional variables. It follows that each element $\mathcal{X}_\mathrm{COM}(M,{\bm E}, {\bm B})$ generates also an admissible postcollisional state.
If $\mathcal{X}_\mathrm{COM}(M,{\bm E}, {\bm B})$ is empty, it means that there are no admissible collision states with the given conserved quantities $(M,{\bm E}, {\bm B})$.
\end{remark}

\begin{remark}
The equation ${\bm B}\cdot{\bm\omega}=0$ in the definition of $\mathcal{X}_\mathrm{COM}(M,{\bm E}, {\bm B})$ 
has a simple geometric meaning. In the spinless relativistic collision problem, the scattering direction is arbitrary on the sphere $S^2$. In the present spinning case, the conservation of the bivector $K$ restricts the scattering direction to the circle $S^2\cap{\bm B}^{\perp}$.
Thus, the additional conservation law reduces the two dimensional family of spinless postcollisional momenta to a one dimensional geometric constraint.
\end{remark}

The non-threshold collision problem in the COM frame has been reduced to the study of the geometry of the manifold $\mathcal{X}_\mathrm{COM}(M,{\bm E}, {\bm B})$. In particular, we show in the next section that the collision manifold is, generically, a finite set. 

\section{Geometry of the collision manifold}\label{geometry}
Our next purpose  is to eliminate the auxiliary variable $\alpha$ in the definition of the collision manifold $\mathcal{X}_\mathrm{COM}(M,{\bm E}, {\bm B})$. 
We first record the exceptional case in which the conserved vectors ${\bm E}$ and ${\bm B}$ are parallel. This case is responsible for the occurrence of infinitely many admissible collision states.

\begin{theorem}\label{exceptionalcase}
Assume that
\[
{\bm E}\times{\bm B}=0.
\]
Define
\begin{equation}\label{CD}
C:=M^2\sigma^2-|{\bm E}|^2-|{\bm B}|^2,
\qquad
D:=\frac{m^2|{\bm B}|^2}{\kappa^2}.
\end{equation}
$\mathrm{(I)}$ If ${\bm B}\neq0$, then the collision manifold satisfies the following alternatives.
\begin{enumerate}
\item If $C<D$, then
\[
\mathcal X_{\mathrm{COM}}(M,{\bm E},{\bm B})=\emptyset.
\]
\item If $C=D$, then
\[
\mathcal X_{\mathrm{COM}}(M,{\bm E},{\bm B})
=
\left\{
({\bm\omega},0)
:
{\bm\omega}\in S^2,\ 
{\bm B}\cdot{\bm\omega}=0
\right\}.
\]
\item If $C>D$, then
\[
\mathcal X_{\mathrm{COM}}(M,{\bm E},{\bm B})
=
\left\{
({\bm\omega},\alpha)
:
{\bm\omega}\in S^2,\ 
{\bm B}\cdot{\bm\omega}=0,\ 
\alpha=\pm\frac1m\sqrt{C-D}
\right\}.
\]
\end{enumerate}
In the special case ${\bm E}={\bm B}=0$ there holds
\[
\mathcal X_{\mathrm{COM}}(M,0,0)=\left\{
({\bm\omega},\alpha)
:
{\bm\omega}\in S^2,\ \ 
\alpha=\pm\frac{M\sigma}{m}
\right\}.
\]
$\mathrm{(II)}$ If $({\bm B}=0,{\bm E}\neq0)$,
then 
\[
C<0\Rightarrow \mathcal X_{\mathrm{COM}}(M,{\bm E},0)=\emptyset,
\]
while for $C\geq 0$ the collision manifold is given by the union
\[
\mathcal X_{\mathrm{COM}}(M,{\bm E},0)=X_1\cup X_2,
\]
where
\[
X_1=
\left\{
({\bm\omega},0):
{\bm\omega}\in S^2,\ 
({\bm E}\cdot{\bm\omega})^2
=\frac{m^2}{\kappa^2}C
\right\},
\]
and
\[
X_2=
\left\{
({\bm\omega},\alpha):
{\bm\omega}\in S^2\cap{\bm E}^{\perp},\ 
\alpha=\pm\frac1m\sqrt{C}
\right\}.
\]
In all cases,
\[
\mathcal X_{\mathrm{COM}}(M,{\bm E},{\bm B})\neq\emptyset
\iff
C\ge D.
\]
Moreover, whenever the collision manifold is nonempty, it contains infinitely many elements.
\end{theorem}

\begin{proof}
$\mathrm{(I)}$ Assume first that ${\bm E}$ is not zero. Then ${\bm E}\times{\bm B}=0$ implies that the vectors ${\bm E}$ and ${\bm B}$ are parallel. Hence, whenever ${\bm B}\cdot{\bm\omega}=0$, we also have ${\bm E}\cdot{\bm\omega}=0$.
Therefore~\eqref{first} is identically satisfied on $S^2\cap{\bm B}^{\perp}$, while~\eqref{second} becomes
\[
\frac{4|{\bm E}|^2}{M^2}+\frac{|{\bm B}|^2}{\kappa^2}+\frac{4m^2}{M^2}\alpha^2
=4\sigma^2.
\]
Multiplying by $M^2/4$ and using
\[
\frac{M^2}{4\kappa^2}=1+\frac{m^2}{\kappa^2},
\]
we obtain
\[
m^2\alpha^2 =
M^2\sigma^2-|{\bm E}|^2-|{\bm B}|^2-\frac{m^2|{\bm B}|^2}{\kappa^2}= C-D.
\]
Therefore
\[
\alpha^2=\frac{C-D}{m^2}.
\]
If instead ${\bm E}$ is zero, then again~\eqref{first} is identically satisfied on $S^2\cap{\bm B}^{\perp}$. Hence, the same conclusion still holds, with the simplification
\[
C-D=M^2\sigma^2-\frac{M^2}{4\kappa^2}|\bm B|^2.
\]
Finally, if ${\bm E}={\bm B}=0$, then $C=M^2\sigma^2$ and $D=0$,
so that
$\alpha=\pm M\sigma/m$,
while ${\bm B}\cdot{\bm\omega}=0$ is vacuous. This yields the stated formula for
$\mathcal X_{\mathrm{COM}}(M,0,0)$.
The result $\mathrm{(I)}$ is proved. 

$\mathrm{(II)}$ Assume ${\bm B}$ is zero, but ${\bm E}$ is not. 
Then~\eqref{first} becomes
\[
({\bm E}\cdot{\bm\omega})\alpha=0.
\]
Hence, either
$
\alpha=0,
$
or
$
{\bm E}\cdot{\bm\omega}=0.
$
Moreover,~\eqref{second} becomes
\[
|{\bm E}|^2
+\frac{\kappa^2}{m^2}({\bm E}\cdot{\bm\omega})^2
+m^2\alpha^2
=M^2\sigma^2,
\]
or equivalently,
\[
\frac{\kappa^2}{m^2}({\bm E}\cdot{\bm\omega})^2
+m^2\alpha^2
=C.
\]
For $C<0$ this equation has no solutions and thus $\mathcal X_{\mathrm{COM}}(M,{\bm E},0)$ is empty. If $C\geq 0$ and $\alpha=0$, then
\[
({\bm E}\cdot{\bm\omega})^2
=\frac{m^2}{\kappa^2}C,
\]
which yields the set $X_1$.
If $C\geq 0$ and ${\bm E}\cdot{\bm\omega}=0$, then
\[
m^2\alpha^2=C,
\]
which yields the set $X_2$.
\end{proof}

\begin{remark}
When ${\bm E}={\bm B}=0$, the conserved bivector $K$ vanishes in the COM frame. This includes the spinless case. In this situation, Theorem~\ref{exceptionalcase} recovers the usual spinless feature that the scattering direction is arbitrary on $S^2$.
\end{remark}
We now turn to the complementary generic case
\[
{\bm E}\times{\bm B}\neq0.
\]
Then ${\bm B}\neq0$, and the condition ${\bm B}\cdot{\bm \omega}=0$ restricts ${\bm \omega}$ to the circle $
S^2\cap{\bm B}^{\perp}$.
Let
\[
{\bm E}_{\perp}:={\bm E}-\frac{{\bm E}\cdot{\bm B}}{|{\bm B}|^2}{\bm B}
\]
be the orthogonal projection of ${\bm E}$ onto the plane ${\bm B}^{\perp}$. Since ${\bm E}\times{\bm B}\neq0$, we have ${\bm E}_{\perp}\neq0$. Define
\begin{equation}\label{basis}
{\bm e}:=\frac{{\bm E}_{\perp}}{|{\bm E}_{\perp}|},\qquad{\bm f}:=\frac{{\bm B}}{|{\bm B}|}\times{\bm e}.
\end{equation}
Then ${\bm e},{\bm f}$ form an orthonormal basis of ${\bm B}^{\perp}$. Hence, every ${\bm \omega}\in S^2\cap{\bm B}^{\perp}$ can be written as
\[
{\bm \omega}=x{\bm e}+y{\bm f},\qquad x^2+y^2=1.
\]
With this parametrization,
\[
{\bm E}\cdot{\bm \omega}=|{\bm E}_{\perp}|\,x.
\]
Moreover, since
\[
{\bm B}\times{\bm E}={\bm B}\times{\bm E}_{\perp}=|{\bm B}|\,|{\bm E}_{\perp}|\,{\bm f},
\]
we have
\[
{\bm \omega}\cdot({\bm B}\times{\bm E})=|{\bm B}|\,|{\bm E}_{\perp}|\,y.
\]
Therefore~\eqref{first} becomes
\[
|{\bm E}_{\perp}|\,x\,\alpha=\frac1\kappa|{\bm B}|\,|{\bm E}_{\perp}|\,y.
\]
Since $|{\bm E}_{\perp}|\neq0$, this is equivalent to
\[
x\alpha=\frac{|{\bm B}|}{\kappa}y.
\]
In particular, $x=0$ is impossible, because then $y=\pm1$ and the right hand side is nonzero. Thus, $x\neq0$, and
\[
\alpha=\frac{|{\bm B}|}{\kappa}\frac{y}{x}.
\]
Substituting this expression in~\eqref{second} we obtain
\[
\frac{4|{\bm E}|^2}{M^2}+\frac{4\kappa^2 |{\bm E}_{\perp}|^2}{M^2m^2}x^2+\frac{|{\bm B}|^2}{\kappa^2}+
\frac{4m^2|{\bm B}|^2}{M^2\kappa^2}\frac{y^2}{x^2}=4\sigma^2.
\]
Multiplying by \(M^2/4\), and using
\[
|{\bm E}_{\perp}|^2=\frac{|{\bm E}\times{\bm B}|^2}{|{\bm B}|^2},
\]
we obtain
\begin{equation}\label{tempo}
A x^2+D\frac{y^2}{x^2}+D=C,
\end{equation}
where $C,D$ are given by~\eqref{CD} and
\begin{equation}\label{A}
A:=\frac{\kappa^2}{m^2}\frac{|{\bm E}\times{\bm B}|^2}{|{\bm B}|^2}.
\end{equation}
Since $x^2+y^2=1$, setting
\[
t:=x^2
\]
gives
\[
0<t\leq1,\qquad y^2=1-t,
\]
and we may rewrite~\eqref{tempo} as the following quadratic equation on $t$:
\begin{equation}\label{Qequation}
Q(t):=A t^2-Ct+D=0.
\end{equation}
Let $\mathcal{Z}$ be the set of solutions to~\eqref{Qequation}.
We have proved the following result.

\begin{theorem}\label{quadraticreduction}
When
\[
{\bm E}\times{\bm B}\neq0
\]
the elements of the collision manifold $\mathcal{X}_\mathrm{COM}(M,{\bm E}, {\bm B})$ are parametrized by the admissible roots of $Q(t)$ as follows. 
\begin{itemize}
\item[(i)] For all $t\in(0,1)\cap\mathcal{Z}$ there exists four elements of  $\mathcal{X}_\mathrm{COM}(M,{\bm E}, {\bm B})$, namely $(\bm\omega,\alpha)=({\bm \omega}_{(\varepsilon,\delta)},\alpha_{(\varepsilon,\delta)})$, 
\[
{\bm \omega}_{(\varepsilon,\delta)}=\varepsilon\sqrt{t}\,{\bm e}+\delta\sqrt{1-t}\,{\bm f},\quad 
\alpha_{(\varepsilon,\delta)}=\frac{|{\bm B}|}{\kappa}\frac{\delta}{\varepsilon}\sqrt{\frac{1-t}{t}},\qquad\varepsilon,\delta\in\{-1,1\},
\]
where $\{{\bm e},{\bm f}\}$ is the orthonormal basis~\eqref{basis} of ${\bm B}^{\perp}$. 
\item[(ii)]
For all $t\in\mathcal{Z}\cap\{1\}$ there exists two elements of  $\mathcal{X}_\mathrm{COM}(M,{\bm E}, {\bm B})$, namely $(\bm\omega,\alpha)=({\bm \omega}_{(\varepsilon)},\alpha_{(\varepsilon)})$
\[
{\bm \omega}_{(\varepsilon)}=\varepsilon{\bm e},\qquad\alpha_{(\varepsilon)}=0,\qquad\varepsilon\in\{-1,1\}.
\]
\item If $\mathcal{Z}\cap (0,1]$ is empty, then so is $\mathcal{X}_\mathrm{COM}(M,{\bm E}, {\bm B})$.
\end{itemize}
\end{theorem}
Figure~\ref{fig:vectors} illustrates the admissible scattering directions introduced in Theorem~\ref{quadraticreduction}. The directions $\bm\omega_{(\epsilon)}$ are obtained from $\bm\omega_{(\epsilon,\delta)}$ in the limiting case $t=1$.

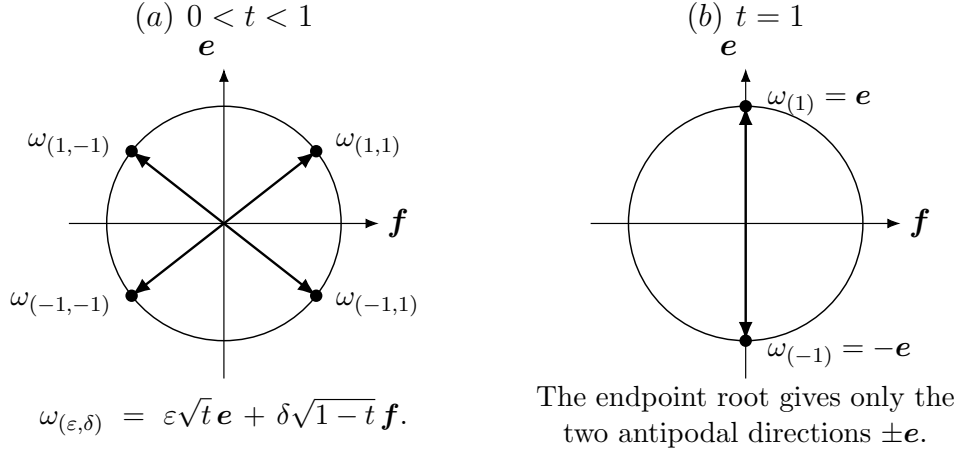
\begin{figure}[ht!]
\centering
\begin{tikzpicture}[
  >=Latex,
  axis/.style={->, line width=0.45pt},
  circleline/.style={line width=0.55pt},
  vec/.style={->, line width=0.9pt},
  point/.style={circle, fill=black, inner sep=1.65pt},
  lab/.style={font=\small},
  title/.style={font=\normalsize\bfseries},
  note/.style={font=\small}
]

\def\R{1.55}
\def\ax{2.05}
\def\troot{0.38}
\pgfmathsetmacro{\x}{sqrt(\troot)}
\pgfmathsetmacro{\y}{sqrt(1-\troot)}

\newcommand{\basecircle}[2]{%
  \begin{scope}[shift={(#1,#2)}]
    \draw[circleline] (0,0) circle (\R);
    \draw[axis] (-\ax,0) -- (\ax,0) node[right=-1pt] {$\bm f$};
    \draw[axis] (0,-\ax) -- (0,\ax) node[above left=1pt and -1pt] {$\bm e$};
  \end{scope}
}

\node[title] at (-3.45,2.72) {$(a)$ $0<t<1$};
\basecircle{-3.45}{0}
\begin{scope}[shift={(-3.45,0)}]
  \draw[vec] (0,0) -- ({\y*\R},{\x*\R});
  \draw[vec] (0,0) -- ({-\y*\R},{\x*\R});
  \draw[vec] (0,0) -- ({\y*\R},{-\x*\R});
  \draw[vec] (0,0) -- ({-\y*\R},{-\x*\R});
  \node[point] at ({\y*\R},{\x*\R}) {};
  \node[point] at ({-\y*\R},{\x*\R}) {};
  \node[point] at ({\y*\R},{-\x*\R}) {};
  \node[point] at ({-\y*\R},{-\x*\R}) {};
  \node[lab, anchor=west] at ({\y*\R+0.12},{\x*\R+0.08}) {$\omega_{(1,1)}$};
  \node[lab, anchor=east] at ({-\y*\R-0.12},{\x*\R+0.08}) {$\omega_{(1,-1)}$};
  \node[lab, anchor=west] at ({\y*\R+0.12},{-\x*\R-0.12}) {$\omega_{(-1,1)}$};
  \node[lab, anchor=east] at ({-\y*\R-0.12},{-\x*\R-0.12}) {$\omega_{(-1,-1)}$};
\end{scope}
\node[note, text width=6.8cm, align=center] at (-3.45,-2.55)
  {$\displaystyle \omega_{(\varepsilon,\delta)}=\varepsilon\sqrt t\,\bm e+\delta\sqrt{1-t}\,\bm f.$};

\node[title] at (3.45,2.72) {$(b)$ $t=1$};
\basecircle{3.45}{0}
\begin{scope}[shift={(3.45,0)}]
  \draw[vec] (0,0) -- (0,\R);
  \draw[vec] (0,0) -- (0,-\R);
  \node[point] at (0,\R) {};
  \node[point] at (0,-\R) {};
  \node[lab, anchor=west] at (0.15,\R+0.04) {$\omega_{(1)}=\bm e$};
  \node[lab, anchor=west] at (0.15,-\R-0.16) {$\omega_{(-1)}=-\bm e$};
\end{scope}
\node[note, text width=6.8cm, align=center] at (3.45,-2.55)
  {The endpoint root gives only the two antipodal directions $\pm\bm e$.};

\end{tikzpicture}
\caption{Admissible scattering directions associated with one admissible root $t\in (0,1]$ of the quadratic polynomial $Q(t)$.}
\label{fig:vectors}
\end{figure}

\begin{remark}
The variable $t=({\bm\omega}\cdot{\bm e})^2$
has a simple geometric interpretation: it is the squared projection of the scattering direction onto the distinguished direction ${\bm e}={\bm E}_{\perp}/|{\bm E}_{\perp}|$
on the collision circle $S^2\cap{\bm B}^{\perp}$. Hence, the quadratic equation $Q(t)=0$ determines the admissible positions of ${\bm\omega}$ on this circle.
\end{remark}
By Theorem~\ref{quadraticreduction}, the cardinality of the collision manifold is determined entirely by the location of the roots of $Q$ relative to the interval $(0,1]$.
To express neatly the number of elements of the collision manifold,
let
\[
N_{\mathrm{int}}=\#\{t\in(0,1):Q(t)=0\},
\]
and let
\[
N_{\mathrm{end}}=
\begin{cases}
1, & Q(1)=0,\\
0, & Q(1)\neq0.
\end{cases}
\]
Then 
\[
{\bm E}\times{\bm B}\neq0 \Rightarrow \# \mathcal{X}_\mathrm{COM}(M,{\bm E}, {\bm B})= 4N_{\mathrm{int}}+2N_{\mathrm{end}}.
\]
Since $Q$ is quadratic, one has $N_{\mathrm{int}}\in\{0,1,2\}$.
Therefore, the possible finite numbers of admissible collision states are
\[
0,\ 2,\ 4,\ 6,\ 8.
\]
Figure~\ref{fig:cardinality} illustrates the possible geometric configurations.

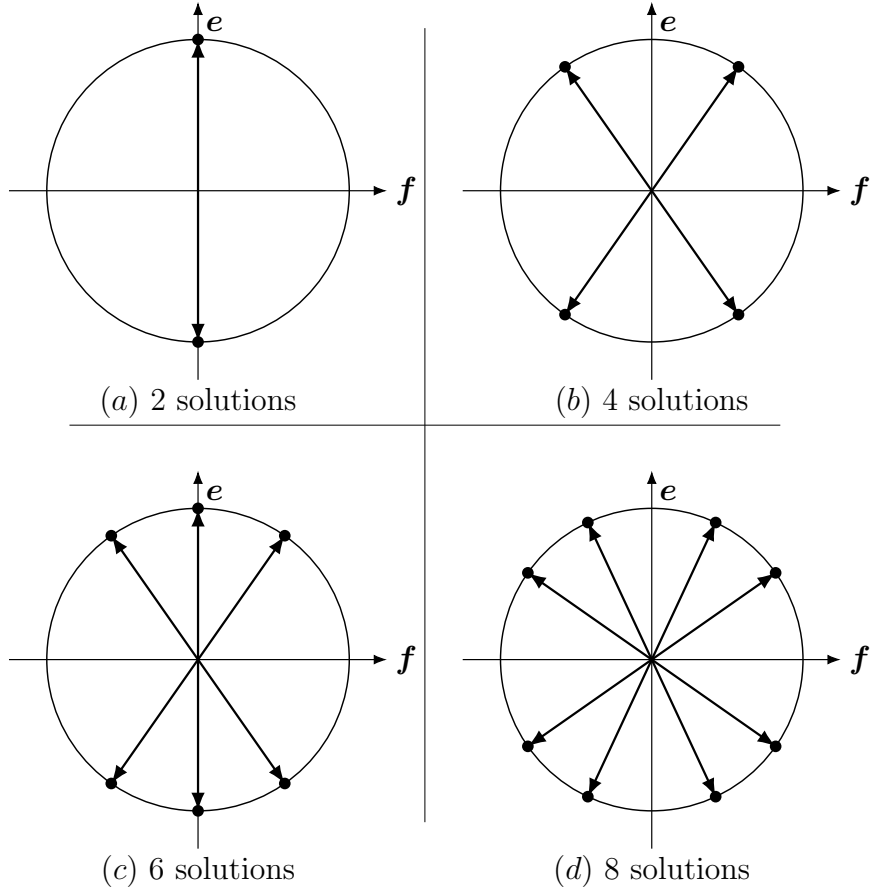
\begin{figure}[ht!]
\centering
\begin{tikzpicture}[
  >=Latex,
  axis/.style={->, line width=0.45pt},
  circleline/.style={line width=0.55pt},
  vec/.style={->, line width=0.85pt},
  point/.style={circle, fill=black, inner sep=1.55pt},
  title/.style={font=\normalsize},
]
\def\R{2}
\def\ax{2.5}
\def\sepX{6}
\def\sepY{6.20}

\newcommand{\panel}[4]{%
  \begin{scope}[shift={(#1,#2)}]
    \node[title] at (0,-2.8) {#3};
    \draw[circleline] (0,0) circle (\R);
    \draw[axis] (-\ax,0) -- (\ax,0) node[right=-1pt] {$\bm f$};
    \draw[axis] (0,-\ax) -- (0,\ax);
    \node at (0.22,2.2) {$\bm e$};
    \foreach \ang in {#4}{
      \draw[vec] (0,0) -- ({\R*cos(\ang)},{\R*sin(\ang)});
      \node[point] at ({\R*cos(\ang)},{\R*sin(\ang)}) {};
    }
  \end{scope}
}

\panel{-\sepX/2}{\sepY/2}{$(a)$ $2$ solutions}{90,270}
\panel{\sepX/2}{\sepY/2}{$(b)$ $4$ solutions}{55,125,235,305}
\panel{-\sepX/2}{-\sepY/2}{$(c)$ $6$ solutions}{90,55,125,235,305,270}
\panel{\sepX/2}{-\sepY/2}{$(d)$ $8$ solutions}{35,65,115,145,215,245,295,325}

\draw[line width=0.35pt] (0,-5.25) -- (0,5.25);
\draw[line width=0.35pt] (-4.7,0) -- (4.7,0);

\end{tikzpicture}
\caption{Typical configurations of the collision manifold
$\mathcal X_{\mathrm{COM}}(M,\bm E,\bm B)$ in the generic case
$\bm E\times\bm B\neq0$. The black arrows represent the admissible
scattering directions $\bm\omega\in S^2\cap\bm B^\perp$. In the of 8 solutions case, there are two roots of $Q(t)$ in the interval $(0,1)$, each giving rise to a set of 4 vectors $\bm\omega_{(\varepsilon,\delta)}$. By Corollary~\ref{nonemptycondition}, the generic
cardinalities of admissible configuration sets are $4$ and $8$.}
\label{fig:cardinality}
\end{figure}

\begin{remark}
The multiplicities in the classification have a direct geometric origin. Each root $t\in(0,1)$ fixes the absolute value of the projection of ${\bm\omega}$ onto ${\bm e}$, but leaves two choices for the sign of this projection and two choices for the sign of the orthogonal component. Thus, each interior root gives four points on the collision circle. In contrast, the endpoint root $t=1$ gives only the two antipodal points $\pm{\bm e}$. 
\end{remark}

In the next theorem we characterize the cardinality of the collision manifold in terms of inequalities on the conserved quantities ${\bm E}$, ${\bm B}$.

\begin{theorem}\label{cardinalitytheorem}
Assume that 
\[
{\bm E}\times{\bm B}\neq0
\]
 and recall the definitions of $A,C,D$ in~\eqref{CD},~\eqref{A}. 
The cardinality of the collision manifold is as follows.

If $D<A$, then
\[
\#\mathcal X_{\mathrm{COM}}(M,{\bm E},{\bm B})
=
\begin{cases}
0, & C<2\sqrt{AD},\\
4, & C=2\sqrt{AD},\ \text{or}\ C>A+D\\
6, & C=A+D,\\
8, & 2\sqrt{AD}<C<A+D.
\end{cases}
\]

If $D\geq A$, then
\[
\#\mathcal X_{\mathrm{COM}}(M,{\bm E},{\bm B})
=
\begin{cases}
0, & C<A+D,\\
2, & C=A+D,\\
4, & C>A+D.
\end{cases}
\]
\end{theorem}

\begin{proof}
For $t>0$~\eqref{Qequation}
is equivalent to
\[
C=At+\frac{D}{t}.
\]
Define
\[
f(t):=At+\frac{D}{t},
\qquad 0<t\le1.
\]
The admissible values of $t$ are precisely the intersections of the graph of $f$ with the horizontal line of height $C$.
Since
\[
f'(t)=A-\frac{D}{t^2},
\]
the unique critical point of $f$ on $(0,\infty)$ is $t_*=\sqrt{D/A}$. 
If $D<A$, then $
t_*\in(0,1)$, $f(t_*)=2\sqrt{AD}$, $f(1)=A+D$
and $f$ has a minimum at $t=t_*$. Hence there are no admissible roots if $C<2\sqrt{AD}$, one interior double root if $C=2\sqrt{AD}$, two interior roots if $2\sqrt{AD}<C<A+D$, one interior root together with the endpoint root $t=1$ if $C=A+D$, and one interior root if $C>A+D$. Since each interior root gives four points of $\mathcal X_{\mathrm{COM}}$, while the endpoint root gives two points, the claimed cardinalities follow.
If $D\geq A$, then $t_*>1$, so $f$ is decreasing on $(0,1]$, with $
f(1)=A+D$.
Thus there are no admissible roots if $C<A+D$, the only admissible root is $t=1$ if $C=A+D$, and there is exactly one interior root if $C>A+D$. This completes the proof.
\end{proof}

\begin{corollary}\label{nonemptycondition}
Assume that ${\bm E}\times{\bm B}\neq0$. Then
\begin{equation}\label{emptyX}
\mathcal X_{\mathrm{COM}}(M,{\bm E},{\bm B})\neq\emptyset
\iff
C\ge
\begin{cases}
2\sqrt{AD}, & D<A,\\
A+D, & D\ge A.
\end{cases}
\end{equation}
Moreover, the cardinality of $\mathcal X_{\mathrm{COM}}(M,{\bm E},{\bm B})$ changes only on the hypersurfaces
\[
C=2\sqrt{AD},
\qquad
C=A+D.
\]
In particular, away from these hypersurfaces the generic finite cardinalities are
\[
\#\mathcal X_{\mathrm{COM}}(M,{\bm E},{\bm B})\in\{0,4,8\}.
\]
\end{corollary}
\begin{remark}
When $A=0$, the inequality in~\eqref{emptyX} reduces to
$
C\ge D,
$
which is precisely the condition obtained in Theorem~\ref{exceptionalcase}
for the collision manifold to be nonempty in the exceptional case
${\bm E}\times{\bm B}=0$.
\end{remark}

\section{Reconstruction of the admissible collision states}\label{reconstructionSEC}
Theorem \ref{cardinalitytheorem} determines the number of admissible collision states for any prescribed conserved quantities. The next theorem shows how to uniquely reconstruct an admissible collision state from each element of the collision manifold, thereby completing the solution of the collision problem.
\begin{theorem}\label{reconstruction}
Assume $\mathcal X_{\mathrm{COM}}(M,{\bm E},{\bm B})\neq\emptyset$. 
Given $(\bm \omega,\alpha)\in\mathcal X_{\mathrm{COM}}(M,{\bm E},{\bm B})$, define the vectors ${\bm U}, {\bm V}$ from the conserved quantities ${\bm E},{\bm B}$ by
\begin{subequations}\label{final}
\begin{equation}\label{final1}
{\bm U}=\frac2M\left({\bm E}-({\bm E}\cdot {\bm \omega}){\bm \omega}\right)+\frac{M}{2m^2}({\bm E}\cdot {\bm \omega}){\bm \omega},\quad
{\bm V}= \frac1\kappa {\bm B}\times {\bm \omega}+\alpha {\bm \omega},
\end{equation}
Then the admissible four-momenta and spin vectors in the COM frame are given by
\begin{equation}\label{final3}
p=(M/2, \kappa\, {\bm \omega}),\quad q=(M/2,-\kappa\,\bm \omega),
\end{equation}
\begin{equation}\label{final4}
{\bm s}=\frac12({\bm U}+{\bm V}),\ s^0=\frac{2}{M}{\bm p}\cdot{\bm s},\quad \quad{\bm r}=\frac12({\bm U}-{\bm V}),\ r^0=\frac{2}{M}{\bm q}\cdot{\bm r},
\end{equation}
and the admissible spin tensors are given by
\begin{equation}
\phi=\frac{1}{\sqrt{2}m}*(p\wedge s),\quad \psi=\frac{1}{\sqrt{2}m}*(q\wedge r).
\end{equation}
\end{subequations}
\end{theorem}
\begin{proof}
The formulas are precisely equations~\eqref{Uprime}, \eqref{Vprime}, \eqref{p},~\eqref{s0},~\eqref{UV} and \eqref{phifroms}. Since the reconstructed variables satisfy the conservation laws together with the mass shell and spin constraints, they define an admissible collision state. The uniqueness is immediate from the reconstruction formulas.
\end{proof}
\begin{remark}
The reconstruction formulas of Theorem~\ref{reconstruction} are expressed in the center of momentum frame, where the collision problem assumes its simplest form. Let $\Lambda$ be the proper orthochronous Lorentz transformation sending the total four-momentum $P$ to $(M,0,0,0)$. If
\[
\langle p_{\mathrm{COM}},s_{\mathrm{COM}}; q_{\mathrm{COM}}, r_{\mathrm{COM}}\rangle
\]
is reconstructed by Theorem~\ref{reconstruction}, then the corresponding collision state in the original Lorentz frame is simply
\[
p=\Lambda^{-1}p_{\mathrm{COM}},\qquad
q=\Lambda^{-1}q_{\mathrm{COM}},
\]
\[
s=\Lambda^{-1}s_{\mathrm{COM}},\qquad
r=\Lambda^{-1}r_{\mathrm{COM}}.
\]
Thus, the reconstruction theorem yields a complete solution of the collision problem in arbitrary inertial frames.
\end{remark}
The reconstruction formulas of Theorem~\ref{reconstruction} are intrinsic,
since they are expressed in terms of the orthonormal basis
$\{{\bm e},{\bm f}\}$ of the plane ${\bm B}^{\perp}$.
For explicit computations, however, it is convenient to exploit the
residual rotational freedom of the COM frame.

Assume that ${\bm E}\times{\bm B}\neq0$.  
Since every spatial rotation preserves the total four-momentum
$(M,0,0,0)$, we may choose the spatial coordinates in the COM frame so
that
\[
{\bm B}=(0,0,B),
\qquad
{\bm E}=(E_{\perp},0,E_{\parallel}),
\]
where
\[
B=|{\bm B}|,
\qquad
E_{\perp}=\frac{|{\bm E}\times{\bm B}|}{|{\bm B}|},
\qquad
E_{\parallel}
=\frac{{\bm E}\cdot{\bm B}}{|{\bm B}|}.
\]
We shall refer to such a coordinate system as an
adapted COM frame.
In these coordinates,
\[
{\bm e}=(1,0,0),
\qquad
{\bm f}=(0,1,0),
\]
and therefore every admissible scattering direction takes the form
\[
{\bm\omega}
=
(x,y,0),
\]
where
\[
x=\varepsilon\sqrt t,
\qquad
y=\delta\sqrt{1-t},\qquad
\varepsilon,\delta\in\{-1,1\},
\]
and $t\in(0,1]$ is an admissible root of the quadratic equation
$Q(t)=0$.
Moreover,
\[
\alpha=
\begin{cases}
\dfrac{B}{\kappa}\dfrac{y}{x}, & 0<t<1,\\[1.2ex]
0, & t=1.
\end{cases}
\]
In the adapted COM frame, the quadratic polynomial $Q(t)$ becomes
\[
Q(t)
=
\frac{\kappa^2E_\perp^2}{m^2}\,t^2
-
Ct
+
\frac{m^2B^2}{\kappa^2},
\]
where
\[
C=M^2\sigma^2-E_\perp^2-E_\parallel^2-B^2.
\]
Moreover, the four-momenta become
\[
p=\left(\frac M2,\kappa x,\kappa y,0\right),
\qquad
q=\left(\frac M2,-\kappa x,-\kappa y,0\right),
\]
while the vectors ${\bm U}$ and ${\bm V}$ simplify to
\[
{\bm U}
=
\left(
\frac{2E_\perp}{M}
+\frac{2\kappa^2E_\perp}{Mm^2}x^2,
\frac{2\kappa^2E_\perp}{Mm^2}xy,
\frac{2E_\parallel}{M}
\right),
\qquad
{\bm V}
=
\left(
0,
\frac{B}{\kappa x},
0
\right).
\]
Consequently,
\[
{\bm s}
=
\frac12({\bm U}+{\bm V}),
\qquad
{\bm r}
=
\frac12({\bm U}-{\bm V}),
\]
that is,
\[
{\bm s}
=
\left(
\frac{E_\perp}{M}
+\frac{\kappa^2E_\perp}{Mm^2}x^2,
\frac{\kappa^2E_\perp}{Mm^2}xy
+\frac{B}{2\kappa x},
\frac{E_\parallel}{M}
\right),
\]
and
\[
{\bm r}
=
\left(
\frac{E_\perp}{M}
+\frac{\kappa^2E_\perp}{Mm^2}x^2,
\frac{\kappa^2E_\perp}{Mm^2}xy
-\frac{B}{2\kappa x},
\frac{E_\parallel}{M}
\right).
\]
Finally, the time components are recovered from
\[
s^0=\frac{2}{M}{\bm p}\cdot{\bm s},
\qquad
r^0=\frac{2}{M}{\bm q}\cdot{\bm r},
\]
and the corresponding spin tensors are given by
\[
\phi=\frac{1}{\sqrt2\,m}*(p\wedge s),
\qquad
\psi=\frac{1}{\sqrt2\,m}*(q\wedge r).
\]

\section{From collision geometry to kinetic theory}\label{kinetic}

The previous sections provide a complete characterization of the collision manifold
\[
\mathcal X_{\mathrm{COM}}(M,{\bm E},{\bm B}),
\]
determined by the conservation laws of energy-momentum and total spin. In particular, Theorems~\ref{cardinalitytheorem} and~\ref{exceptionalcase} show that, in the generic case
\[
{\bm E}\times{\bm B}\neq0,
\]
the collision manifold contains at most eight elements, whereas in the exceptional case it is either empty or infinite.

This observation has important consequences for kinetic theory. Indeed, in the classical Boltzmann theory of spinless particles, the conservation of momentum and energy determines the unit sphere of scattering directions in the center of mass frame. The collision operator is therefore obtained by integrating over this continuous family of admissible postcollisional states.
 By contrast, for relativistic particles with spin, the conservation laws determine a generically finite collision manifold whose cardinality is completely determined by the conserved quantities. Consequently, the continuous angular degree of freedom of the classical theory is replaced by a discrete collision geometry, and the angular integration appearing in the Boltzmann collision operator is naturally replaced by a finite summation over the admissible collision states.

To construct a kinetic collision operator, however, the conservation laws alone are not sufficient. One must additionally prescribe how transitions occur between admissible collision states. In the classical Boltzmann equation this information is encoded by the differential cross section, or equivalently by the collision kernel. Motivated by this analogy, we introduce the following notion.

\begin{definition}
Assume that ${\bm E}\times{\bm B}\neq0$. 
A collision kernel is a nonnegative function
\[
W:
\mathcal X_{\mathrm{COM}}(M,{\bm E},{\bm B})
\times
\mathcal X_{\mathrm{COM}}(M,{\bm E},{\bm B})
\longrightarrow
[0,\infty),
\]
where $W(x,y)$ denotes the transition rate from the collision state
$x\in\mathcal X_{\mathrm{COM}}(M,{\bm E},{\bm B})$
to the collision state
$y\in\mathcal X_{\mathrm{COM}}(M,{\bm E},{\bm B})$.
\end{definition}

The collision kernel constitutes additional physical information that is not determined by the conservation laws.  Depending on the underlying microscopic model, one may further require properties such as Lorentz covariance, microscopic reversibility or suitable normalization conditions. A deterministic collision process corresponds to the special case in which, for every collision state $x$, the transition kernel is concentrated on a single collision state.

The finiteness of the collision manifold suggests that the collision operator in the relativistic Boltzmann equation should be expressed as a finite sum over admissible collision states rather than as an integral over scattering directions. Schematically, the angular integration
\[
\int_{S^2}(\cdots)\,d\omega
\]
of the classical Boltzmann equation is replaced by the finite summation
\[
\sum_{y\in\mathcal X_{\mathrm{COM}}(M,{\bm E},{\bm B})}
W(x,y)(\cdots).
\]
This discrete collision geometry provides a natural geometric foundation for a relativistic kinetic theory of spinning particles. 
The construction and analysis of the corresponding relativistic Boltzmann equation for relativistic spinning particles, including its conservation laws, equilibrium states and entropy properties, will be the subject of a subsequent publication.

\end{document}